\documentclass[11pt]{article}

\usepackage{authblk}
\usepackage{amsmath,amssymb}
\usepackage{epsfig, graphicx,fleqn,xspace,color,graphics}

\let\OLDthebibliography\thebibliography
\renewcommand\thebibliography[1]{
  \OLDthebibliography{#1}
  \setlength{\parskip}{0pt}
  \setlength{\itemsep}{0pt plus 0.3ex}
}

\usepackage{helvet,graphicx,makeidx,multicol}

\usepackage{float}
\usepackage{setspace}
\usepackage{cite}

\usepackage{xcolor}

\usepackage{algorithm}
\usepackage{framed}
\usepackage{algpseudocode}
\newfloat{alg}{thp}{lop}
\floatname{alg}{Algorithm}

\usepackage{xspace}
\DeclareMathOperator*{\argmin}{argmin}

\newenvironment{proof}{\noindent{\bf Proof.} }{\null\hfill$\Box$\par\medskip}

\newtheorem{theorem}{Theorem}[section]

\newtheorem{lemma}{Lemma}[section]

\newtheorem{corollary}{Corollary}[section]

\newtheorem{definition}{Definition}[section]

\newcounter{smallitemizec}

\date{}
\begin{document}
\title{Distributed Pattern Formation in a Ring\footnote{This research was supported by NSERC of Canada}}
\author[1]{Anne-Laure Ehresmann}
\author[2]{Manuel Lafond}
\author[1]{Lata Narayanan} 
\author[1]{Jaroslav Opatrny} 

\affil[2]{Department of Computer Science, University of Sherbrooke, Canada}
\affil[1]{Department of Computer Science and Software Engineering, Concordia University, Montreal, Canada}
\affil[ ]{\textit {\{alehresmann@gmail.com,Manuel.Lafond@USherbrooke.ca,
lata@encs.concordia.ca,opatrny@cs.concordia.ca\}}}
\maketitle

\begin{abstract}
Motivated by concerns about diversity in social networks, we consider the following  pattern formation problems in rings. Assume $n$ mobile agents are located at the nodes of an $n$-node ring network. Each agent is assigned a colour from the set $\{c_1, c_2, \ldots, c_q \}$. The ring is divided into $k$ contiguous {\em blocks} or neighbourhoods of length $p$. The agents are required to rearrange themselves in a distributed manner to satisfy given diversity requirements:  in each block $j$ and for each colour $c_i$, there must be exactly $n_i(j) >0$ agents of colour $c_i$ in block $j$. Agents are assumed to be able to see agents in adjacent blocks, and move to any position in adjacent blocks in one time step. 

When the number of colours $q=2$,  we give an algorithm that terminates in time $N_1/n^*_1 + k + 4$ where $N_1$ is the total number of agents of colour  $c_1$ and $n^*_1$ is the minimum number of agents of colour $c_1$ required in any block. When the diversity requirements are the same in every block, our algorithm requires $3k+4$ steps, and is asymptotically optimal.  Our algorithm generalizes for an arbitrary number of colours, and terminates in $O(nk)$ steps.  We also show how to extend it to achieve arbitrary specific final patterns, provided there is at least one agent of every colour in every pattern. 
\end{abstract}

\section{Introduction}

Recent research in sociology and network science indicates that diverse social connections have many benefits. For example,  de Leon {\em et al}~\cite{Leon2017} conclude that increasing diversity, and not just increasing size, of social networks may be essential for improving health and survival among the elderly. Eagle {\em et al} find that social network diversity is at the very least a strong structural signature for the economic development of communities \cite{Eagle1029}. Reagans and Zuckerman \cite{Reagans01} found that increased organizational tenure diversity in R\&D  teams
correlated positively with higher creativity and productivity. 

On the other hand, the pioneering work of Schelling \cite{schelling1969models,schelling1971dynamic} provided a model to describe how even small preferences for {\em locally} homogeneous neighbourhoods result in globally segregated cities. In his model, individuals of two colours are situated on a path, ring or mesh (i.e. a one- or two-dimensional grid) and have a threshold for the minimum number of neighbors of the same colour they require in their local neighborhood. If this threshold is not met, they move to a random new location. Schelling showed via simulations that this process inevitably led to a globally segregated pattern. This model has been studied extensively and the results confirmed repeatedly; see for example  \cite{benard2007wealth,benenson2009schelling,dall2008statistical,henry2011emergence,pancs2007schelling,young2001individual,zhang2004dynamic}.   Recent work in theoretical computer science \cite{brandt2012analysis,immorlica2017exponential} has shown that the expected size of these segregated communities can be exponential in the size of the local neighbourhoods in some social network graphs.   The Schelling model has also been studied from a games perspective, where neighborhood formation is determined by agents that can be selfish, strategic~\cite{chauhan2018schelling,elkind2019schelling} or form coalitions~\cite{bredereck2019hedonic}.  

In this paper, we study an algorithmic approach to {\em seeking} diversity. Consider a social network formed by a finite number of individuals, henceforth called {\em agents} that can be classified into a set of categories, called {\em colours}. The group collectively seeks diversity in {\em local}  neighborhoods. Is it possible to achieve a specified version of diversity? If so, how can the agents achieve this diversity?
These questions were first raised and studied in a recent paper \cite{latin2018}, in which the authors studied a model with red and blue agents specifying their local neighbourhood preferences in a ring network.  Centralized algorithms to satisfy the preferences of all agents, when possible, were given in \cite{latin2018}. 

In this paper, we study distributed algorithms for achieving diversity in local neighborhoods (blocks) in a  ring network.

\subsection{Model and problem definition}
We assume  we are given  a collection of $n$ \textit{agents} situated on the nodes of  an $n$-node ring network. 
We fix a  partition of the ring  into $k$  paths or {\em blocks} of length $p$, with $n=kp$. 
Each agent  has a colour drawn from the set $\{c_1, c_2, \ldots, c_q \}$. For ease of exposition, we initially focus our analysis on the case $q=2$, and call the colours $c_1$ and $c_2$ as {\em blue}  and {\em red} respectively.  Any specific clockwise 
ordering of these $n$ coloured agents around the ring  is called a 
\textit{configuration} or {\em $n$-configuration}. Starting from an initial configuration, we are interested in distributed algorithms for the agents to rearrange themselves into final configurations that meet certain given constraints, for example, alternating red and  blue agents, or each red agent has at least
one  blue agent as a neighbour. 

A configuration can be represented by a  string ${\cal C}$ of length $n$ drawn from the alphabet $\{c_1, c_2, \ldots, c_q \}$. Following standard terminology, for any string $u$, we denote by $|u|$ the length of $u$; by $u^i$ the string $u$ repeated $i$ times; by $uv$ the string $u$ concatenated with the string $v$. The number of occurrences of the colour $c_i$ in a string $u$ is denoted by $n_i(u)$.  A {\em pattern} $P$ is a string of length $p$ drawn from the alphabet $\{c_1, c_2, \ldots, c_q \}$ with $n_i(P) > 0$ for all $1 \leq i \leq q$. Corresponding to the partition of the ring into blocks mentioned above, a configuration can be seen as a concatenation of $k$ patterns $S_1S_2 \ldots S_k$, in which $S_i$ is the pattern of agents in the $i$-th block. 

We are interested  in distributed algorithms for the following problems:

\begin{enumerate}
\item[P1:] Given positive integers $n_i(j)$ for all $1 \leq i \leq q$ and $1 \leq j \leq k$ and a valid initial configuration ${\cal C}$, give a distributed algorithm that achieves a final configuration $S_1S_2\ldots S_k$ with
$n_i(S_j) = n_i(j)$ for all $1 \leq i \leq q$ and $1 \leq j \leq k$. 

\item[P2:] Given positive integers $n_i(j)$ for all $1 \leq i \leq q$ and $1 \leq j \leq k$ and a valid initial configuration ${\cal C}$, give a distributed algorithm that achieves a final configuration $S_1S_2\ldots S_k$ with
$n_i(S_j) \geq n_i(j)$ for all $1 \leq i \leq q$ and $1 \leq j \leq k$. 
\item[P3:]  Given a sequence of patterns $P_1, P_2, \ldots ,P_k$, and a valid initial configuration ${\cal C}$, give a distributed algorithm to achieve the final configuration 
$P_1P_2 \cdots P_k$.
\end{enumerate}
A given initial configuration ${\cal C}$ is {\em valid} for the three problems above (respectively) if and only if  for all $1 \leq i \leq q$, we have (P1) $n_i({\cal C}) = \sum_{j=1}^{k} n_i(j)$  
(P2) $n_i({\cal C}) \geq \sum_{j=1}^{k} n_i(j)$ and (P3) $n_i({\cal C}) = \sum_{j=1}^{k} n_i(P_j)$. Clearly the problem is only solvable for valid initial configurations. 
Figure~\ref{fig:example} shows an input configuration and an output configuration for problem P1 with the required number of agents of colours $c_1$ and $c_2$ being $2$ in all the blocks.

Notice that an algorithm for P1 can be used to solve P3 with a few modifications.
If specific patterns are required in blocks as in Problem P3, we can first apply
an algorithm for problem P1 
to obtain
the correct number of agents of different colours in every block. Once a block has the correct
number of agents of every colour, they can rearrange themselves in one additional step to form
the required pattern.

\begin{figure}[h]
\label{fig:example}
    \centering
\includegraphics[width=0.47\textwidth]{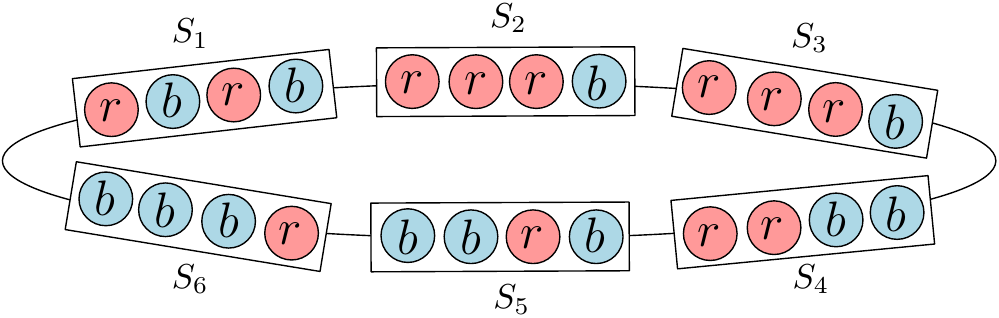}
\includegraphics[width=0.47\textwidth]{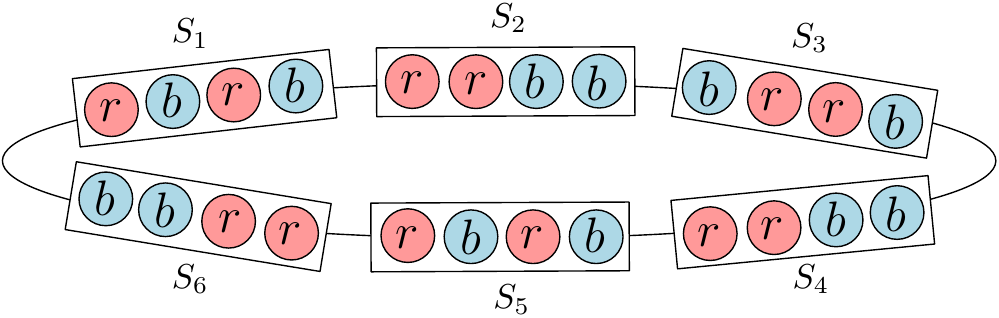}
\caption{An input configuration on the left, and a possible output on the right.}
\end{figure}

\noindent {\bf Agent model.} Our agent model is similar to the widely studied {\em autonomous mobile agent} model \cite{FPS2012}, where initially the agents are assumed to occupy arbitrary positions in the plane, and each agent repeatedly performs a {\em Look-Compute-Move} cycle. First an agent looks at the positions of the other agents. Then,
using the positions of other agents it found and its own, it computes its next position. Finally it moves to the newly computed position. The agents are generally assumed to have identical capabilities, they have no centralized coordination, and they do not communicate with each other. Indeed, their decisions are based only on their observations of their surroundings made during the look phase. Many different variations of the model have been studied, based on whether or not the agents are synchronized or not, how far they can see, whether or not they agree on a coordinate system, etc.

In this paper, as already stated, we assume that agents are initially placed at the nodes of a $n$-node ring network $G$. They belong to $q$ colour classes, but agents within a colour class are identical, and they do not have or use any id. Agents are aware of their position in the ring, and may use this position in their calculations; for example agents at odd positions in the ring may behave differently from agents at even positions.  They are completely synchronous, and  time proceeds in discrete time steps. We assume that agents have a common {\em limited visibility range}: each agent can see all agents in the ring that are within the same block or in adjacent blocks. They also have a limited movement range: in one step, an agent can move to any position in the same or adjacent blocks\footnote{We note that if the visibility and movement ranges are smaller, our algorithms still work, though they take more time. The details are lengthy and uninteresting, and therefore omitted from this paper.}. It is the responsibility of the algorithm designer to ensure that there are no collisions, {\em i.e.},  we must never have two agents in the same position in the ring. Agents are memoryless, and do not communicate; their decisions are made solely on what they see during their``look" phase. We say that a distributed algorithm for pattern formation by agents {\em terminates} in a time step, if agents have formed the required pattern, and no agent moves after this time step. 


\subsection{Our results}

We give solutions for problems P1, P3, and a restricted case of P2. We first consider in detail the case when the number of blocks $k$ is even, and the number of colours $q=2$, that is, all agents are blue or red.  We give an algorithm that terminates in time $N_b/n^*_b + k + 4$ where $N_b$ is the total number of  blue agents, and $n^*_b$ is the minimum number of blue agents  required in any block (Theorem~\ref{thm:time-bound}). Note that this time is upper bounded by $n/2$ in the worst case where $n$ is the total number of agents. When the requirements are homogeneous, that is, the same number of blue agents are required in every block, our algorithm terminates in $3k+4$ steps; we show that this is optimal (Corollary~\ref{cor:same}). 

We then give  algorithms for the cases when:
\begin{enumerate}
\item $k$ is odd.
\item the number of colours $q$ is more than 2. 
\item  a specific final configuration  $P_1 P_2\ldots P_k$ is desired.
\end{enumerate}
or any combination of the above (see Theorem~\ref{th:termination-q}). These algorithms terminate in $O(nk)$ time.

Finally, given $n_1(j)>0$ for all $1 \leq j \leq k$, we give an $O(nk)$ algorithm to achieve a final configuration in which the number of agents of colour $c_1$ is at least $n_1(j)$ in every block $S_j$ (see Theorem~\ref{th:termination-lower-bounded}).

\subsection{Related work}

There is a large literature on distributed computing by autonomous mobile agents, see the book \cite{FPS2012} that provides a very complete survey of results in this area. Much of this work is done for mobile agents located in the plane or a continuous region, and the main problems that have been
considered are gathering of robots \cite{peleg2005,FPSW05}, scattering and covering \cite{peleg2008}, and pattern formation \cite{SY99,cice16}
by autonomous mobile robots (see~\cite[section 1.3]{flocchini2019distributed}). The solvability, i.e., termination or convergence, of these problems depends on the type of synchronization, visibility, agreement on orientation \cite{FPS2012}. 
In~\cite{michail2016simple}, the authors introduce the problem of using a set of pairwise-communicating processes to constructing a given network.

In  discrete spaces, typically graphs like grids or rings, similar problems have been studied, see for example: gathering \cite{klasing2008,dangelo2012}; exploration  \cite{lamani2010}; pattern formation \cite{BAKA}; and deployment \cite{elor2011}.  Unlike in our paper, in all these studies, only a subset of the nodes of the graph is occupied by agents, and typically agents are identical, and not partitioned into colour classes.


The work closest to our work is \cite{latin2018}, which considers a similar setting of red and blue agents on a ring of $n$ nodes, and every node of a ring contains an agent. The agents are required to achieve final configurations on a ring, with specified diversity constraints similar to ours. However, unlike the present paper, in which the final configuration specifies the diversity constraints per block and agents are interchangeable,
in \cite{latin2018}, each agent has a preference for the colours of
 other agents  in its $w$-neighborhood. The authors consider there three ways for agents to specify these preferences: each agent can specify  (1) a preference list:  the {\em sequence} of colours of agents in the neighborhood, (2) a preference type: the {\em exact} 
number of neighbors of its own colour in its neighborhood,  or (3) a preference threshold: the {\em minimum }number of agents of its own colour in its neighborhood. The main result of \cite{latin2018} is that satisfying seating preferences is 
fixed-parameter tractable (FPT) with respect to parameter $w$ for preference types and thresholds, while it can be solved in $O(n)$ time for preference lists. They also show a linear time algorithm for the case when all agents have homogeneous preference types. However, in \cite{latin2018},  all algorithms are centralized unlike in the present paper, in which we are only interested in distributed algorithms.

\section{Distributed Algorithm for Two colours} 
\label{sec:dist-algo}

In this section  we consider problem P1 with $q=2$, that is, all agents are blue or red. We denote the number of blue and red agents in a pattern $u$ by $n_b(u)$ and $n_r(u)$ respectively. We first describe an algorithm for the first problem, {\em i.e.} when the exact number of blue and red agents in the  $i$-th block, denoted by $n_b(i)$ and $n_r(i)$ respectively,  is specified. 
For simplicity, we assume the number of blocks $k$ is even. We will show in Section~\ref{odd-case} that the algorithm also works when $k$ is odd. Let $N_b$ and $N_r$ be the total number of blue and red agents in the initial configuration. Furthermore, let 
$$n_b^* = \min_{1\leq j \leq k} n_b(j) \hspace*{0.3in}\mbox{    and    } \hspace*{0.3in}n_r^* = \min_{1\leq j \leq k} n_r(j)$$
We  assume without loss of generality that $N_b/n_b^* \leq N_r/n_r^*$. If not, we reverse the roles of the blue and red agents in the algorithm and subsequent analysis.

%
We use $S_i$ to denote both the $i$-th block and the pattern of agents  in it at a particular time step; the meaning used will be clear from context.  For every $i$, $1\leq i\leq k$ we denote by $y_i$ the difference between the number of 
blue agents in  $S_i$ and the desired number of blue agents:
$$y_i= n_b(S_i)-n_b(i)$$
We call $y_i$ a surplus  if positive, and a deficit if negative. 
  
 In each step, we partition the ring into \emph{windows} of length $2p$, each window consisting of 2 adjacent blocks. 
If blocks $S_i$ and $S_{i+1}$ are paired into a window in this partitioning, we denote the window by $[S_i | S_{i+1}]$.  Recall that in our agent model, agents can see and move to adjacent blocks, and know their positions in the ring. Therefore, agents in the same window can all see each other and can move to any position within the window (as $S_i$ has access to $S_{i+1}$ and vice-versa). Also, notice that this pairing is specific to our algorithm, and not a constraint of our model.  Agents in the window $[S_i | S_{i+1}]$ execute the following algorithm

%
%

\begin{algorithm}
\begin{algorithmic} 
\If{$y_i< 0$, i.e., $S_i$ has a deficit of blue agents and an excess of red agents}
	\State $t \gets \min (n_b^*, |y_i|, n_b(S_{i+1}))$
    \State Move blue agents before red agents in order-preserving manner in $S_i$ and $S_{i+1}$,
	\State  $t$  leftmost red agents of $S_i$  swap places with  $t$ leftmost blue  agents of $S_{i+1}$.
\EndIf
\end{algorithmic}
\caption{Algorithm for agents in the window $[S_i | S_{i+1}]$} 
\label{alg:window-algo}
\end{algorithm}

Note that neither moving blue agents before red agents, nor the specification that it is the leftmost agents that swap places, is necessary for the algorithm to terminate, nor do these affect the termination time; they only simplify the analysis. Furthermore, agents can compute the net effect of the two moves mentioned above, and perform a single movement in a single  
Look-Compute-Move cycle; thus
Algorithm~\ref{alg:window-algo} comprises a single time step.

  
Clearly agents in different windows can perform the above algorithm in parallel in the same time step without collisions. We give our pattern formation algorithm for agents of two colours below:  

\begin{algorithm}
\begin{algorithmic}
\State $i \gets 1$ 
\Loop{}

        \State Apply in parallel Algorithm  \ref{alg:window-algo} to windows
              $[S_i|S_{i+1}],  [S_{i+2}|S_{i+3}],  \ldots, [S_{i-2},S_{i-1}]$ 
\State $i \gets 1 +i \mod n$
\EndLoop
\end{algorithmic}
\caption{Pattern Formation Algorithm  }
\label{alg:main-algo-s}
\end{algorithm}

We call one iteration of the loop in Algorithm~\ref{alg:main-algo-s} a {\em round}. It is straightforward to see that each round takes constant time and can be performed without collisions. 
Notice that if a block is paired with the block on its right in a round, then it will be paired with the block on its left in the next round. 
In one round, surplus red agents (if any)  in the blocks  $S_i,S_{i+2}, S_{i+4}, \ldots, S_{i-2}$ move to  $S_{i+1},S_{i+3}, S_{i+5},$ $\ldots,S_{i-1}$  respectively, 
and in the next round surplus red agents in the blocks $S_{i+1},S_{i+3}, S_{i+5}, \ldots, S_{i-1}$ move to $S_{i+2},S_{i+4}, S_{i+6}, \ldots, S_i$,
respectively, and this is repeated.  
In other words, there is a flow of red agents to the right and an equivalent flow of blue agents to the left.  Since blue agents moving left from the block $S_{i+1}$ to $S_i$ are always the leftmost blue agents in $S_{i+1}$, and move to positions after the blue agents already present in $S_i$, and the blue agents move up to the beginning of each block before the next round,  each round preserves the clockwise order of the blue agents 
 as stated in the following lemma:

\begin{lemma} 
[Order-preserving property of blue agents] \label{lem:order} Let ${\cal C}$ and ${\cal C'}$ be the configurations before and after a round in Algorithm~\ref{alg:main-algo-s}. Then the clockwise ordering of blue 
agents in ${\cal C}$ and ${\cal C'}$ is the same. 
\end{lemma}

\subsection{Proof of correctness and termination}

We now show that this flow of agents eventually stops, and the algorithm terminates with the desired number of blue and red agents in every block. We first introduce some definitions and observations about the properties of a configuration. 

\begin{figure}[h]
\label{fig:lemma-excess}
    \centering
    \includegraphics[width=0.485\textwidth]{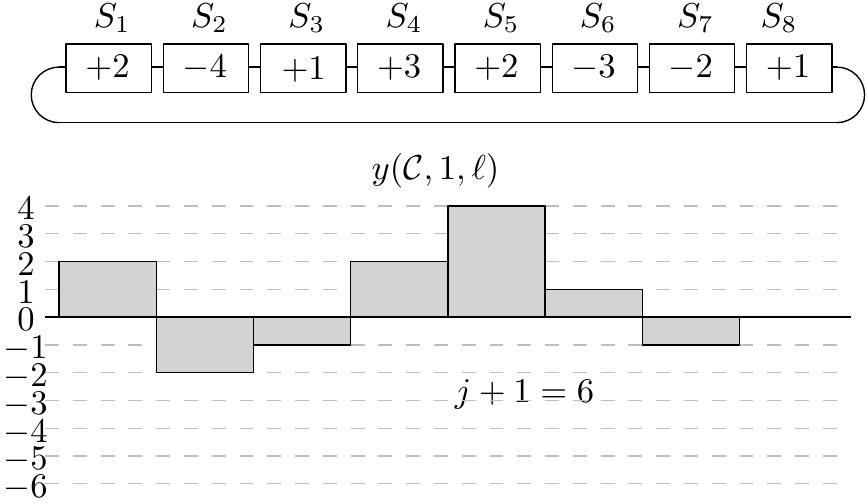}
    \includegraphics[width=0.485\textwidth]{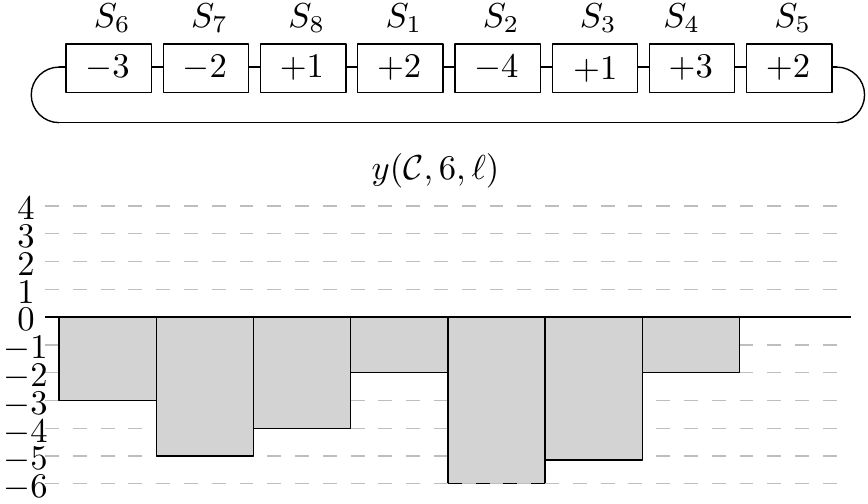}
    \caption{The surplus of blue agents $y_i$ is shown inside the box for slice $S_i$. The plot below shows the cumulative surplus $y({\cal C}, 1, \ell)$ on the left and $y({\cal C}, 6, \ell)$ on the right}
\end{figure}

\begin{definition} 
For any given valid configuration ${\cal C}=S_1S_2\ldots S_k$, define $y({\cal C},i,\ell)$ $=\sum_{j=0}^{\ell-1}{y_{i+j}}$, i.e.,  $y_i$ is the cumulative surplus of blue agents in  $\ell$ consecutive blocks starting from $S_i$.
\end{definition}

\begin{definition}

For any given valid configuration ${\cal C}=S_1S_2 \ldots S_k$, define $y({\cal C},i)$ to be the maximum in the set $\{y({\cal C},i,1), y({\cal C},i,2), \ldots,y({\cal C},i,k)\}$. 
\end{definition}

\noindent Figure~\ref{fig:lemma-excess} shows $y({\cal C}, 1, \ell)$ for a sample configuration. Notice that $y({\cal C}, 1) = y({\cal C}, 1, 5) = 4$.

\begin{lemma} 
\label{lemma:excess} In any given valid configuration ${\cal C}=S_1 S_2 \ldots S_k$ there exists an integer $i$ such that $y({\cal C},i)=0$.
\end{lemma}
\begin{proof}
Let $j$ be the integer such that 
$y({\cal C},1,j)=y({\cal C},1)$, i.e., the index of the block where 
the sum of differences of number of blue agents between the blocks and patterns, when counted from $S_1$, is maximized. See Figure~\ref{fig:lemma-excess}. Consider 
the value $y({\cal C},j+1,\ell)$ for some integer $\ell$. Clearly,   $y({\cal C},j+1,\ell) =  y({\cal C},1,\ell)-y({\cal C},1,j)\leq 0$ by our choice of $j$. On the other hand in any valid 
configuration the total excess of red agents in the configuration is $0$. 
Thus, $y({\cal C},j+1)=0$. See Figure~\ref{fig:lemma-excess}; in this example, \ $j=5$; the figure on the right shows $y({\cal C}, 6, \ell)$, which is always $\leq 0$ for every $\ell$.
\end{proof}

Since the input configuration is circular, and the algorithm runs in parallel in all blocks, we can rename the blocks. We will therefore assume in the sequel that  $y({\cal C},1)=0$. This means that starting in block $S_1$, {\em the cumulative surplus of blue agents is never positive.} This implies the following useful non-positive prefix and non-negative suffix property. 

\begin{lemma} [Non-Negative Suffix Property] \label{lemma:pref-suf} For any valid input configuration $\cal C$ and for every integer $j$, $1\leq j < k$ we have:
$$y({\cal C},1,j)\leq 0 \mbox{ and }y({\cal C},j+1,k-j)\geq 0$$
In other words, the total surplus of blue agents in the first $j$
blocks is at most $0$ while the total surplus of blue agents in the last $k-j$
blocks is at least $0$. 
\end{lemma}
\begin{proof}
Since $y({\cal C},1)=0$, we have that $y({\cal C},1,j)\leq 0$. 
Since  $y({\cal C},1)=y({\cal C},1,j)+y({\cal C},j+1,k-j)= 0$, we get 
$y({\cal C},j+1,k-j)=- y({\cal C},1,j)\geq 0$. 
\end{proof}

In particular, notice that the non-negative suffix property implies that  $y_k=y({\cal C},k,1)\geq 0$.We now show that  every round of Algorithm  \ref{alg:main-algo-s} maintains the non-negative suffix property.  

\begin{lemma}\label{lemma:invariant}
Let  ${\cal C}=S_1S_2\ldots S_k$ be a configuration such that 
 for any $j$, $1 \leq j \leq k$, we have  $y({\cal C},1,j)\leq 0$ and  $y({\cal C},j+1,k-j)\geq 0$.
Executing one round of  Algorithm  \ref{alg:main-algo-s} 
gives a configuration  ${\cal C'}=S_1'S_2' \ldots S_k'$ such that  $y({\cal C'},1,j)\leq 0$ and  $y({\cal C'},j+1,k-j)\geq 0$ for any $j$, $1 \leq j \leq k$.
\end{lemma}

\begin{proof}
For a block $S_\ell$, define $z_\ell = n_r(S_\ell) - n_r(\ell)$ as the surplus of red agents in $S_\ell$.
Let $\ell$ be an integer, $1\leq \ell<k$ such that Algorithm \ref{alg:window-algo} is applied to the pair $[S_\ell|S_{\ell+1}]$.  Assume for now that $\cal C'$ is the configuration obtained after only $[S_\ell|S_{\ell+1}]$ have made their exchanges and no other blocks.  We will show that the conditions of the Lemma hold for $\cal C'$.  This proves our statement, since we may repeat our analysis by applying the exchanges between pairs of blocks one at a time while maintaining the invariants of the Lemma.

Performing the exchanges between $S_\ell$ and $S_{\ell + 1}$ can affect 
only the suffixes $y({\cal C'},\ell,k-\ell+1)$ and $y({\cal C'},\ell+1,k-\ell)$; any  movement of agents that happens cannot affect a suffix $y({\cal C'}, j, k-j)$ with
$j < \ell$ or $j > \ell+1$. Clearly, if  $y(S_\ell) \geq 0$ or there are no blue agents in $S_{\ell+1}$ then there is no movement of agents between $S_\ell$ and $S_{\ell+1}$ and thus, obviously, 
  $y({\cal C'},\ell,k-\ell+1)=y({\cal C},\ell,k-\ell+1)\geq 0$, and  $y({\cal C'},1,\ell-1)=y({\cal C},1,\ell-1)\leq 0$.
  
Thus we only need to consider the case when   $S_\ell$ has $z_\ell= - y_l>0$ excess red agents and $S_{\ell+1}$ contains $n_b(S_{\ell+1})>0$  blue agents. 
Let $t=min\{z_\ell,n_b(S_{\ell+1}),n_b^*\}$. Then   $t$ blue agents move from $S_{\ell+1}$ to $S_\ell$ and 
  $t$ move  from $S_{\ell}$ to $S_{\ell+1}$. Clearly,  the total surplus of blue agents in the $\ell$-th suffix does not change, that is, $y({\cal C'},\ell,k-\ell+1)=y({\cal C},\ell,k-\ell+1)\geq 0$.
We only need to show that the movement of  $t$ blue agents to $S_\ell$ does not render the  $(\ell+1)$-th suffix negative. Observe that 
$$0\leq y({\cal C},\ell,k-\ell+1)=
y({\cal C},\ell+1,k-\ell) -z_\ell $$

In other words, since $S_\ell$ lacks $z_\ell$ blue agents and yet the $\ell$-th suffix is non-negative, it must be that the $(\ell+1)$-st suffix is at least $z_\ell$. Thus, $y({\cal C},\ell+1,k-\ell)\geq z_\ell \geq t$, and 
moving $t \leq z_\ell$ blue agents from $S_{\ell+1}$ to $S_\ell$ cannot make the $(\ell+1)$-st suffix negative. That is, 
 $$y({\cal C'},\ell+1,k-\ell)= y({\cal C},\ell+1,k-\ell)- t \geq 0$$
\end{proof}

Lemma \ref{lemma:invariant} implies that the property  $y_k\geq 0$  is maintained after every round. This means $S_k$ {\em never} has a deficit of blue agents, and therefore when it participates in the window $[S_k|S_1]$, there is no swap of agents. Therefore we obtain the following corollary:
\begin{corollary}
\label{cor:last-sl}

Algorithm \ref{alg:main-algo-s} never does any swap of agents between the blocks
 $S_k$ and $S_1$.
\end{corollary}

We now develop a potential function argument to show that Algorithm~\ref{alg:main-algo-s} eventually reaches the desired final configuration, after which no agents ever move.
Given a configuration
$ {\cal C}=S_1S_2\cdots S_k$, we define the {\em destination block} of the $i$-th blue agent (from position 1 in the ring) 
to be
$$dest(i) =  \argmin_{\ell} \sum_{i=1}^\ell n_i(\ell) \geq i$$
That is, the destination block of the  leftmost $n_b(1)$  blue agents in the input configuration is $S_1$, the next $n_b(2)$ blue agents in the input configuration is $S_2$, and so on. Next we define the 
{\em displacement} of the $i$-th blue agent (from position 1 in the ring) that is located in block $S_j$ to be the difference between the index of its present block and the index of its destination block:
$$displacement(i, j) = j - dest(i)$$
Given a configuration $ {\cal C}=S_1S_2\cdots S_k$, we define the {\em distance} of the configuration $d( {\cal C})$  to be the sum of displacements of all blue agents in the configuration. We first prove some properties of the distance function.

\begin{lemma}
\label{potential-properties}
Let ${\cal C}$ be the configuration after $i$ rounds in Algorithm~\ref{alg:main-algo-s}. Then:
\begin{enumerate}
\item $d({\cal C}) \geq 0$
\item $d({\cal C}) \leq d({\cal C'})$ where $\cal C'$ was the configuration before the $i$-th round.
\item If $d({\cal C}) = 0$, no agent ever moves again. 
\end{enumerate}

\end{lemma}

\begin{proof}
It follows from Lemma~\ref{lemma:invariant} that the displacement of any blue agent is non-negative, as otherwise, there would be a deficit of blue agents in some suffix. Thus $d({\cal C})$ must always be non-negative. Second, since blue agents only move left in an order-preserving manner (Lemma~\ref{lem:order}), and from Corollary~\ref{cor:last-sl}, they never move left past  $S_1$ to $S_k$, the displacement of any blue agent cannot increase, and therefore $d({\cal C})$ cannot increase.  Finally, if $d({\cal C}) = 0$, then every block has its desired number of blue agents, and  no agent will ever move again.
\end{proof}

We now show that the distance function is guaranteed to decrease  after at most 2 rounds.



\begin{lemma} 
\label{lemma:displ}
Assume that we are given a configuration  ${\cal C} = S_1 S_2 \ldots, S_k$ and positive integers   $n_b(i)>0$ for $1\leq j
 \leq k$ such that $d({\cal C})>0$. After at most two rounds in Algorithm \ref{alg:main-algo-s} to ${\cal C}$ we obtain configuration ${\cal C}'$ such that 
$d({\cal C}')\leq d({\cal C})-1$. 
\end{lemma}
\begin{proof}
Since  $d({\cal C})>0$, we know that ${\cal C}$ is not a valid final configuration, and Lemma \ref{lemma:pref-suf} and the assumption $n_b(i)>0$ for all $1\leq j \leq k$ implies that there exists a window  $[S_j | S_{j+1}]$ such that $n_b(S_j)<n_b(j)$ and $n_b(S_{j+1})>0$, i.e., 
there is a surplus  of red agents in $S_j$, and $S_{j+1}$ contains 
some blue agents.   
Then $t=\min\{|y_j|,n_b(S_{j+1}),n_b^*\}$ as defined in Algorithm~\ref{alg:window-algo} is positive. 
Now, if the window  $[S_j | S_{j+1}]$ is active in the next round in Algorithm \ref{alg:main-algo-s}, then $t$ blue agents  from $S_{j+1}$ are swapped 
with $t$  red agents from $S_j$. Clearly, the displacement of these  blue agents goes down by 1. Thus after this round, we obtain  
configuration ${\cal C'}$ with  $d({\cal C'})<d({\cal C})$. 

Suppose instead the window  $[S_j | S_{j+1}]$ is not active in the next round of Algorithm \ref{alg:main-algo-s}. Let ${\cal C}_1$ be the configuration obtained after the next round of Algorithm \ref{alg:main-algo-s}. By Lemma~\ref{potential-properties}, $d({\cal C}_1) \leq d({\cal C})$.  
In the following round of Algorithm \ref{alg:main-algo-s} the window  $[S_j | S_{j+1}]$ is  active and we transfer $t$ blue agents from $S_{j+1}$ to  $S_j$, and the displacement of these  $t$ agent goes down by 1. Thus after two rounds, we obtain the configuration ${\cal C}'$ with  $d({\cal C}')\leq d({\cal C}_1)-1 \leq d({\cal C})-1$. 
\end{proof}
Notice that the condition  $n_b(P_j)>0$ for $1\leq j \leq k$ is necessary.  For example, if there is a block $S_j$ with $n_b(P_j)=0$ and $y(S_j)=0$ then this block will never receive a blue agent, and this would prevent any flow of blue agents to the left of $S_j$, even when blue agents  are needed there, and any flow of red agents 
across $S_j$.   

\begin{theorem} 
\label{th:termination}

Given $n_b(j)>0$ and $n_r(j) \geq 0$ for all all $1 \leq j \leq k$, where $k$ is even, and a valid input configuration, Algorithm~\ref{alg:main-algo-s} reaches a configuration where $n_b(S_j) = n_b(j)$ and $n_r(S_j) = n_r(j)$ in every block $S_j$. Furthermore, this takes $O(nk)$ steps, and after this configuration is reached, no agent ever moves again.

\end{theorem}
\begin{proof}
Since the initial value of the distance function is at most $nk$, the result follows from Lemma~\ref{lemma:displ} and Lemma~\ref{potential-properties}.
\end{proof}

\subsection{Proof of termination time} \label{sec:analysis}
In this section, we show an upper bound on the number of iterations required for Algorithm~\ref{alg:main-algo-s} to terminate. 
Recall that $n_b^* = p - \max_{1 \leq j \leq k}  (n_r(j))$ and $N_b$ is the total number of blue agents in the input configuration ${\cal C}$.
Let $B$ be the set of blue agents, and  assume (for the purpose of the analysis) that they are numbered from left to right as $b_1, b_2, \ldots b_{N_b}$. Partition the members of $B$ into subsets $B_1, \ldots, B_{L}$ as follows.
Let $B_1 = \{b_1, \ldots, b_{n_b^*}\}$, $B_2  = \{b_{n_b^* + 1}, \ldots, b_{2n_b^*}\}$, and so on, so that 
$L = \lceil N_b/n_b^* \rceil$ and each subset has $n_b^*$ elements, except perhaps $B_L$. 
We say that an agent $b_i$ is \emph{$(t)$-cooperative} if, in every round $t' \geq t$, either 
(1) $b_i$ moves left; or (2) $b_i$ has reached its destination block as defined in the previous section.

The key idea  is that the first $n_b^*$ blue agents (i.e. $B_1$) start travelling every round at the latest by round $2$ until they reach their destination blocks. 
After this round, they ``free'' the path for the blue agents of $B_2$, which then start travelling every round from round $4$ until they reach their destination blocks. 
This unblocks $B_3$ from round $6$, and so on.  In the end, the blue agents  of $B_L$ can freely travel from round  $2L + 2$, 
and they only need to travel at most $k$ slices.  Note that agents in blocks $B_l$ may indeed move before round $2l + 2$ but they may be subsequently blocked; however from round $2l+2$ onward, they will continue to travel every round until they reach their destination block. These ideas are formalized in the next lemma.

\begin{lemma}\label{lem:that-technical-lemma}
For every $l \in [L]$, each blue agent in $B_l$ is $(2l + 2)$-cooperative. 
\end{lemma}

\begin{proof}
For convenience, define $B_0 = \emptyset$.  
By induction over $l \in \{0, 1, \ldots, L\}$ we show that each blue agent in each $B_l$ is $(2l + 2)$-cooperative. 
Using $l = 0$ as a base case, this statement is vacuously true: each blue agent in $B_0$ is $(0)$-cooperative.
Now suppose that the lemma holds for all $B_{l'}$ with $l' < l$.  
Assume that there is $b_i \in B_l$ such that in some round $w \geq 2l + 2$, the agent $b_i$ 
does not move left and it is not in its destination block. 
Let $S_a$ be the block that contains $b_i$ before (and after) the execution of round $w$.
It follows from the earlier discussions that $S_a \neq S_1$, as this would mean the destination block of $b_i$ requires it to move left from $S_1$ to $S_{k}$.

There are 2 cases: either, in round $w$, 

\vspace*{0.1in}

(*) $S_a$ is in the active window $[S_{a-1} | S_a]$;

\vspace*{0.1in}
(**) $S_a$ is in the active window $[S_a | S_{a+1}]$.

\vspace*{0.1in}

Let us first suppose that (*) holds. Note that since $w \geq 2l + 2$, we may assume by induction that all the blue agents  in  $B_1 \cup \ldots B_{l-1}$ are $(w- 1)$-cooperative\footnote{The reader may bear in mind that the proof works in the particular case $l = 1$ and $w= 1$}.
This means that at the start of round $w$,  the block $S_{a-1}$ cannot contain any blue agent in $B_1 \cup \ldots \cup B_{l-1}$, as such an agent could not advance in round $w$, contradicting its $(w-1)$-cooperativeness. 
Moreover, by Lemma~\ref{lem:order}, $S_{a-1}$ also cannot contain any blue agent in $B_{l+1}, \ldots, B_L$, as this would break the ordering of the blue agents since $b_i \in B_l$ is in $S_a$.
It follows that $S_{a-1}$ can only have  blue agents from $B_l$ at the start of round $w$.

Suppose $S_{a-1}$ has  $z>0$ blue agents from $B_l$ at the beginning of round $w$. Then it has $p-z$ red agents. Since  the maximum number of red agents required in a block is  
$p-n_b^*$, there are at least $p-z-(p-n_b^*)=n_b^*-z$ agents that  are surplus red agents in $S_{a-1}$. Meanwhile by the order preserving property of Lemma~\ref{lem:order}, $S_a$ contains no agents from $B_j$ for any $j < l$, since $S_{a-1}$ has an agent from $B_l$. Furthermore, $S_a$ contains at most $n_b^*-z$ blue agents from $R_l$, and any agents that it contains from $B_j$ with $j>l$ would occur to the right of any blue agents from $R_{l}$. Hence agent $b_i$ would move to $S_{a-1}$ in round $w$, a contradiction. 

If instead $S_{a-1}$ has no blue agents at the beginning of round $w$, it must have only red agents, that is, $p$ of them.  Then  there must be at least $n_b^*$ excess red agents in $S_{a-1}$, and at most $n_b^*$ agents from $B_l$ in $S_a$. Thus the only reason that $r_i$  would not move to $S_{a-1}$  in round $w$ is that there are agents from $B_{l-1}$ in $S_{a}$ which have priority moving left over the agent $b_i$. Let $b_h$ be such an agent.  Consider the round $w- 1$.
At the start of that round, $S_a$ was in the window $[S_{a} | S_{a+1}]$ and so $S_a$ could not have contained a blue agent from $B_1 \cup \ldots B_{l-1}$, as these would not be able to move left in round $w-1$, contradicting our inductive hypothesis.  Consequently, $b_h$ was in $S_{a+1}$ and by the order preserving property of Lemma~\ref{lem:order}, $S_a$ could not contain a blue agent  from $B_l \cup \ldots \cup B_L$ either.  In other words $S_a$ had no blue agents at the start of round $w - 1$.  
 At most $n_b^*$ blue agents could have moved left to $S_a$ in round $w- 1$, and it follows that at the start of round $w$, $S_a$ has at most $n_b^*$ blue agents.
But as we argued, $S_{a-1}$ has at least $n_b^*$ excess red agents at the start of round $w$, which means that 
all the blue agents in $S_a$ are able to move left to $S_{a-1}$, including $b_i$, a contradiction.

Next let us suppose that (**) holds. At the start of round $w - 1$, the agent $b_i$ was either in $S_{a}$ or in $S_{a+1}$. 
But if $b_i$ was in $S_{a+1}$ at the start of round $w-1$,   it could not move from $S_{a+1}$ to $S_a$, as  $S_{a+1}$ was in the window $[S_{a+1} | S_{a+2}]$ in that round.  
Thus $b_i$ was in $S_a$ in round $w - 1$ and it did not advance, noting that $S_a$ belonged to the window $[S_{a-1}|S_{a}]$.  
We now use an identical argument for time $w-1\geq 2l+1$, as we did above for the case when (*) holds, to obtain a contradiction. We  conclude that every $b_i \in B_l$ is $(2l + 2)$-cooperative.\end{proof}

\begin{theorem}
\label{thm:time-bound}
Given even $k$ and positive integers $n_r(j)$ and $n_b(j)$ for all $1\leq j \leq k$ and a valid input configuration ${\cal C}$, Algorithm~\ref{alg:main-algo-s} terminates in at most $\frac{2 N_b}{n_b^*} + k + 4$ time steps, where $n_b^* = \min_{1 \leq j \leq k}  (n_b(j))$ and $N_b$ is the total number of blue  agents in ${\cal C}$.
\end{theorem}

\begin{proof}
Consider the partition $B_1, \ldots, B_L$ of the blue agents as defined above.
We have $L = \lceil N_b /n_b^*\rceil \leq N_b / n_b^* + 1$.
By Lemma~\ref{lem:that-technical-lemma}, each $b_i$ is $(2L + 2)$-cooperative, implying that each $b_i$ is also $(2 N_b / n_b^* + 4)$-cooperative.
This means that after at most $2 N_b / n_b^* + 4$ rounds, every blue agent that is not yet at its destination block advances a block in every round. 
Since there are $k$ blocks, such a $b_i$ will reach its destination after at most $k$ additional rounds.
Since this holds for any $b_i$, and recalling that every round can be executed in one step, this proves the statement.  
\end{proof}

\begin{corollary}
\label{cor:same}
If the diversity requirements are homogeneous, that is $n_b(i) = m$ for all $1 \leq j \leq k$, then Algorithm~\ref{alg:main-algo-s} terminates in at most $3k  + 4$ steps, and this is  optimal. \end{corollary}

\begin{proof}
Suppose $n_b(j) = m$ for all $1 \leq j \leq k$. Then $n_b^*= m$ and $N_b = km$,
and  the bound given by Theorem~\ref{thm:time-bound} is $\frac{2 km}{m} + k + 4 = 3k + 4$.

We now show that there are inputs for which $\Omega(k)$ is a lower bound on the number of steps required for {\em any} algorithm in our agent model. 
 In particular consider an input with the size of the blocks $p$ even, and $n_b(i) = p/2$ for every block $i$, and the input configuration such that the first $k/2$ blocks have only red agents and the last $k/2$ blocks have only blue agents. Then $pk/4$ blue agents need to move to blocks $S_1$ to $S_{k/2}$. Since agents can only move to the next block in a single time step, all these agents would have to move from $S_{k/2+1}$ to $S_{k/2}$ or from $S_{k}$ to $S_{1}$ in some step. However, at most $p$ blue agents can move from $S_{k/2+1}$ to $S_k$ and similarly at most $p$ blue agents can move from $S_k$ to $S_1$ in one time step. It follows that the number of time steps required is at least $pk/(4(2p)) = k/8$.
\end{proof}

Notice that when the diversity  patterns are identical, given by  a
string $u$ of length $p$, the solution obtained by our algorithm is of type $ u^k$. 
The results in \cite{latin2018} imply that in this case, in the final ring network, in any neighbourhood of length $p$ starting at any position of the ring, the  number of blue and red agents is the same as in $u$.

\subsection{When the number of blocks $k$ is odd}

\label{odd-case}

Consider now the situation when the number of blocks $k$ is odd. 
Notice that in one round of Algorithm~\ref{alg:main-algo-s}, the surplus red agents from the  blocks $S_i,S_{i+2}, S_{i+4},\\ \ldots, S_{i-3}$ move to $S_{i+1},S_{i+3}, S_{i+5}, \ldots, S_{i-2}$,
respectively, and in the next round, surplus red agents from the blocks $S_{i+1},S_{i+3}, S_{i+5}, \ldots S_{i-2}$ move  to   $S_{i+2},S_{i+4},\\ S_{i+6},\ldots,S_{i-1} $ respectively. In this case surplus red agents from $S_{i-1}$  can move only in the subsequent third round, but this delay happens to  $S_{i-1}$ only once in $k$ rounds since the value of $i$ is incremented after every round.  
Then in the proof of Lemma \ref{lemma:displ} the window 
$[S_i,S_{i+1}]$ in which there is swap of agents between $S_i$ and $S_{i+1}$ is only in the third round of the algorithm. Thus,  the value of $d({\cal C})$ is only guaranteed to decrease after 3 rounds. However, due to the way the windows are shifted in the algorithm after each round, this can occur to $S_{i-1}$ only once in $k$ rounds. It is clear that Theorem \ref{th:termination} still holds.

\section{Distributed Algorithm for $q$ colours}

In this section we consider the pattern formation problem  when the given set of colours $\{c_1,c_2,\ldots,c_q\}$ is of size  $q\geq 3$. We show that Algorithm \ref{alg:main-algo-s} can be generalized to solve the pattern formation problem for any $q\geq 3$.

The main idea of the generalization is to iteratively "correct" the number of agents of colour $i$, starting with colour $c_1$. Assume that we are given an input instance ${\cal C}$ consisting of a set of colours $\{c_1,c_2,\ldots,c_q\}$, $q\geq 3$,  number $n_i(j)$ for each colour specifying the required number of agents of colour $c_i$ in block $S_j$, and a valid input configuration  ${\cal C}$ over alphabet $\{c_1,c_2,\ldots,c_q\}$.  

 Let $c_2'$ be a new colour not in $\{c_1,c_2,\ldots,c_q\}$.  Consider now the  input instance
${\cal C}_1$ of two colours $c_1$ and $c_2'$, obtained from ${\cal C}$  by treating all  colours
$\{c_2,\ldots,c_q\}$ as $c_2'$.  For every $1 \leq j \leq k$, the parameters $n_1(j)$ remain the same as in the original instance,  while  $n'_2(j) = \sum_{2 \leq i \leq q} n_i(j)$, that is, the required number of agents of colour $c_2'$ in block $S_j$ is the sum of requirements of colours $\{c_2,\ldots,c_q\}$ in block $S_j$ of the original problem. If we run Algorithm \ref{alg:main-algo-s} on the transformed input ${\cal C}'$ we get a final configuration ${\cal C}_1'$  in which the number of agents of colour $c_1$ in each block is correct. 

Next we  consider input instance
 ${\cal C}_2$ of two colours $c_2$ and $c_3'$  obtained from ${\cal C}_1'$  by treating all colours $\{c_3,\ldots,c_q\}$ as a new colour $c_3'$. The requirement $n_2(j)$ remains as in ${\cal C}$ while  $n_3'(j)$ is the sum of requirements of colours $\{c_3,\ldots,c_q\}$ in block $S_j$ of the original problem.
  If we run Algorithm \ref{alg:main-algo} on input instance ${\cal C}_2$ while the agents of colour  $c_1$ do not participate or move, we get a final configuration ${\cal C}_2'$  in which the number of agents of colours $c_1$ and $c_2$ in each block is correct. 
Clearly this can be repeated in $q-1$ phases to get a solution to the original instance ${\cal C}$; in phase $i$, we obtain the correct number of agents of colour $i$ in every block.  The above description assumes that the phases are run sequentially, and phase $i+1$ starts after phase $i$ is completed. 

In fact, since  in phase $i$, the agents of colours $\ell < i$ do not move, and as such their numbers do not change,  it is enough for agents to verify {\em locally} that the number of agents of colours $c_\ell$ with $\ell < i$ is correct before applying the algorithm for phase $i$. In other words, the phases can be interleaved. We describe how to modify the algorithms for a window $[S_j,S_{j+1}]$ as given in Algorithm \ref{alg:window-algo-q}. 


\begin{algorithm}
  \begin{algorithmic}
 \State $i \gets 1$
\While {$i < q \And (n_i(S_j)=n_i(j)\And n_i(S_{j+1})=n_i(j+1))$}   
\State $i \gets i+1$
 \EndWhile
 \If {$i < q$ }
\If{$n_i(S_j) <n_i(j)$, i.e., $S_j$ needs additional agents of colour $c_i$}
	\State $t \gets \min (n_i(j)-n_i(S_j), n_i(S_{j+1}))$
	\State  $t$   agents of colour $c_i$ in $S_{j+1}$ swap places with   $t$ agents of colour  greater than $i$ in $S_{j}$
        \EndIf
\Else {~} Rearrange $S_i$ into specific pattern $P_i$ if required in the problem specification. 
\EndIf
\end{algorithmic}
\caption{Algorithm for a Window $[S_j | S_{j+1}]$ for Agents of  $q$ colours} 
\label{alg:window-algo-q}
\end{algorithm}

\begin{algorithm}
\begin{algorithmic} 
\State $i \gets 1$
\Loop{}
    \State Apply in parallel Algorithm  \ref{alg:window-algo-q} to windows
             $[S_i|S_{i+1}],  [S_{i+2}|S_{i+3}] \ldots$
\State $i \gets 1+i\mod n$
\EndLoop
\end{algorithmic}
\caption{Pattern Formation Algorithm for Agents of $q$ colours } 
\label{alg:main-algo}
\end{algorithm}

Clearly, the termination of this  algorithm follows directly from the termination of
Algorithm \ref{alg:main-algo} as shown in Theorem \ref{th:termination}. Notice that $n_i(c_j)\geq 1$ is needed for every $1\leq i \leq q-1$ and $1 \leq j \leq k$ to guarantee termination.


The discussion above yields the following theorem:

\begin{theorem} 
\label{th:termination-q}
Algorithm ~\ref{alg:main-algo} solves  the  pattern formation problems P1 and P3 and terminates in time $O(nk)$ where $n$ is the number of agents and $k$ is the number of blocks.
%
%
%

\end{theorem}

\subsection{When there is a lower bound on the number of agents} 
\label{sec:bounded}
In this section we address a restricted version of the second problem P2:  given 
$n_1(j) > 0$ and $n_i(j) = 0$ for all $2 \leq i \leq q$ and  $1 \leq j \leq k$, achieve a final configuration ${\cal C'}=S_1S_2 \ldots S_k$ with $n_i(S_j) \geq n_i(j)$. In other words, there is a positive lower bound required on the number of agents of colour $c_1$ in every block, but no lower bound on agents of any other colour. Observe  that in this case a valid initial configuration $w$ must satisfy the 
condition $n_1({\cal C}) \geq \sum_{i=1}^k n_1(S_i)$.
 
Let $d = n_1({\cal C}) -\sum_{i=1}^k n_1(S_i)$, that is, $d$ is the number of extra agents of colour $c_1$ present in the input configuration. 
Consider now applying Algorithm \ref{alg:main-algo-s} (for 2 colours) on the given input instance, treating agents of all colours except for colour $c_1$ as agents of colour $c_2'$, a new colour.   
Given a valid block configuration ${\cal C}$ we use the definitions of $y({\cal C},i,\ell)$ and   $y({\cal C},i,\ell)$ from
the preceding section. It is easy to see that  by replacing every $0$ by $d$ in the statements of Lemmas 
\ref{lemma:excess}, \ref{lemma:pref-suf}, \ref{lemma:invariant}, and Corollary \ref{cor:last-sl}, they 
remain valid for this problem.  
This implies that Lemma \ref{lemma:displ} and Theorem \ref{th:termination} remain valid exactly as stated in the previous section. Thus we get the following:

\begin{theorem} 
\label{th:termination-lower-bounded}
Algorithm ~\ref{alg:main-algo} solves the following restricted version of the pattern formation problem P2 in time $O(nk)$: Given a valid input configuration ${\cal C}$ and 
 $n_1(j) > 0$ and $n_i(j) = 0$ for all $2 \leq i \leq q$ and  $1 \leq j \leq k$, achieve a final configuration ${\cal C'}=S_1S_2 \ldots S_k$ with $n_i(S_j) \geq n_i(j)$.

\end{theorem}

\section{Discussion} \label{discussion}

Given $n$ agents of $q$ different colours situated on the nodes of a ring network that has been partitioned into $k$ blocks, we gave distributed algorithms for the agents to move to new locations where they satisfy specified patterns or diversity constraints in every block. Our analysis for the case of even $k$ and two colours is tight, but for the other problems, it would be interesting to obtain tight bounds. Our algorithms need a positive number of agents of at least the first $q-1$ colours to be required in the final configuration in every block; it is unclear to what extent this restriction is necessary for the existence of distributed algorithms for the problem.  

Although we only considered a ring network, our algorithms could be adapted to 
other neighbourhood networks provided that a ring-like flow among the neighbourhoods can be established.   
It would also be interesting to consider bounds on the movement capacity or bandwidth
and its impact on the number of steps needed.

\bibliographystyle{plain}
\bibliography{ref}

\end{document}